\documentclass[]{interact}
\usepackage{epstopdf}%
\usepackage{url}
\usepackage{enumitem}

\usepackage{subcaption}

\usepackage{natbib}
\bibpunct[, ]{(}{)}{;}{a}{}{,}
\renewcommand\bibfont{\fontsize{10}{12}\selectfont}

\numberwithin{equation}{section}

\usepackage{hyperref}

\usepackage{xcolor} 
\usepackage{mathtools}

\theoremstyle{plain}
\newtheorem{theorem}{Theorem}[section]
\newtheorem{lemma}[theorem]{Lemma}
\newtheorem{corollary}[theorem]{Corollary}

\theoremstyle{definition}
\newtheorem{definition}[theorem]{Definition}
\newtheorem{example}[theorem]{Example}

\theoremstyle{remark}
\newtheorem{remark}{Remark}

\newcommand{\Prob}{\mathsf{P}}
\newcommand{\Expect}{\mathsf{E}}
\DeclareMathOperator*{\esssup}{ess\,sup}

\def\rc{\textcolor{black}}
\def\rcn{\textcolor{black}}

\begin{document}


\title{Robust Quickest Change Detection in Non-Stationary Processes}

\author{
\name{Yingze Hou\textsuperscript{a}, Yousef Oleyaeimotlagh\textsuperscript{a}, Rahul Mishra\textsuperscript{c},  Hoda Bidkhori\textsuperscript{b}, Taposh Banerjee\textsuperscript{a}\thanks{CONTACT: Taposh Banerjee. Email: taposh.banerjee@pitt.edu}}
\affil{\textsuperscript{a}University of Pittsburgh, \textsuperscript{b}George Mason University, \textsuperscript{c}Indian Space Research Organization.
}
}
\maketitle

\begin{abstract}
\rcn{Exactly and asymptotically optimal algorithms are developed for robust detection of changes in non-stationary processes. In non-stationary processes, the distribution of the data after change varies with time. The decision maker does not have access to precise information on the post-change distribution. It is shown that if the post-change non-stationary family has a distribution that is least favorable in a well-defined sense, then the algorithms designed using the least favorable laws are robust optimal. This is the first result where an exactly robust optimal solution is obtained in a non-stationary setting, where the least favorable law is also allowed to be non-stationary. Examples of non-stationary processes encountered in public health monitoring and space and military applications are provided. Our robust algorithms are also applied to real and simulated data to show their effectiveness.}

\end{abstract}


\section{Introduction}
In classical quickest change detection (QCD) theory,
optimal and asymptotically optimal solutions are developed to detect a sudden change in the distribution of a stochastic process. The strongest results are available for the independent and identically distributed (i.i.d.) setting.  In the QCD problem in the i.i.d. setting, a decision maker observes a sequence of random variables $\{X_n\}$. 
Before a time $\nu$ (called the change point), the random variables are i.i.d. with a fixed density $f$, and after $\nu$, are i.i.d. with another density $g$:
\begin{equation}\label{eq:iidQCD}
	X_n \sim
	\begin{cases}
		f, &\quad \forall n < \nu, \\
		g, &\quad \forall n \geq \nu.
	\end{cases}
\end{equation}
Here $f$ and $g$ are densities such that
$$
D(g \; \| \; f) := \int g(x) \log \frac{g(x)}{f(x)} dx \; > \; 0.
$$
The goal of the QCD problem is to detect this change in distribution from $f$ to $g$ with the minimum possible delay subject to a constraint on the rate of false alarms (\cite{veer-bane-elsevierbook-2013, tart-niki-bass-2014, poor-hadj-qcd-book-2009}). 

In the Bayesian setting, it is assumed that the change point is a random variable.
The optimal solution is the Shiryaev test (\cite{shir-siamtpa-1963, tart-veer-siamtpa-2005}):
\begin{equation}
    \label{eq:iidShir}
    \tau_s = \min\{n \geq 1: \Prob(\nu \leq n | X_1, \dots, X_n) \geq A\}.
\end{equation}
The asymptotic optimality of this test is established in \cite{tart-veer-siamtpa-2005}. This test is exactly optimal when the change point is a geometrically distributed random variable: $\nu \sim \text{Geom}(\rho)$. In this case, the Shiryaev statistic has a simple recursion : if $p_n = \Prob(\nu \leq n | X_1, \dots, X_n)$ is the Shiryaev statistic, then the statistic $R_n$ defined as $R_n = \frac{p_n}{1-p_n}$ can be written as
\begin{equation}
 R_n = \frac{1}{(1-\rho)^n}\sum_{k=1}^n (1-\rho)^{k-1}\rho \; \prod_{i=k}^n \frac{g(X_i)}{f(X_i)}, 
\end{equation}
and has the simple recursion: 
\begin{equation}
    \label{eq:ShirRecur}
    R_n = \frac{R_{n-1} + \rho}{1-\rho} \; \frac{g(X_n)}{f(X_n)}, \quad R_0=0.
\end{equation}
The exact problem for which this is the optimal solution is discussed in Section~\ref{sec:Bayes}. 
However, we note that the threshold $A$ must be carefully selected to satisfy a constraint on the probability of a false alarm. 

In non-Bayesian settings (\cite{poll-astat-1985, lord-amstat-1971}), the change point $\nu$ is treated as an unknown constant. In this case, the notion of an average delay is not defined. A notion of conditional delay (conditioned on the change point) can be defined, but the average delay value depends on the location of the change point. Thus, a minimax approach is taken. 
See Section~\ref{sec:minimax} for a precise mathematical formulation. 
The optimal solution for the i.i.d. setting is given by the cumulative sum (CUSUM) algorithm (\cite{page-biometrica-1954, mous-astat-1986, lai-ieeetit-1998}): 
\begin{equation}
    \label{eq:iidcusum}
    \tau_c = \min\{n \geq 1: W_n \geq A\}.
\end{equation}
Here the CUSUM statistic $W_n$ is given by
\begin{equation}
    W_n = \max_{1 \leq k \leq n} \sum_{i=k}^n \log \frac{g(X_i)}{f(X_i)},
\end{equation}
and also has an efficient recursion:
\begin{equation}
    W_n = \left(W_{n-1} + \log \frac{g(X_n)}{f(X_n)}\right)^+, \quad \quad W_0=0.
\end{equation}
The threshold $A$ in \eqref{eq:iidcusum} must be chosen carefully to meet a constraint on the false alarm rate.
For a review of the QCD literature, we refer to  \cite{bansal1986algorithm, lai-ieeetit-1998, tart-veer-siamtpa-2005, tart-niki-bass-2014, tart-book-2019}.

In many practical change detection problems (see Section~\ref{sec:motivation} for details), e.g., detecting approaching debris in satellite safety applications, detecting an approaching enemy object in military applications, and detecting the onset of a public health crisis (e.g., a pandemic), the change point model encountered differs from the model discussed above in two fundamental ways:
\begin{enumerate}
    \item The observation variables may still be independent, but the post-change model may not be stationary, i.e., the data density after the change varies with time. Specifically, a more common model encountered in practice is
    \begin{equation}\label{eq:changepointmodel_0}
	X_n \sim
	\begin{cases}
		f, &\quad \forall n < \nu, \\
		g_{n}, &\quad \forall n \geq \nu.
	\end{cases}
      \end{equation}
      With a slight abuse of notation (borrowing it from the time-series literature [12]), we will refer to a process with densities $\{g_{n}\}_{n \geq \nu}$  as a non-stationary process since the density of the data evolves or changes with time. 
      \item The post-change density information is not precisely known. For example, the densities $\{g_{n}\}$ are not available to the decision maker. 
\end{enumerate}

The main aim of this paper is to address these two issues and develop robust optimal solutions for change detection in the non-stationary setting of \eqref{eq:changepointmodel_0}. In Section~\ref{sec:Bayes}, we discuss the problem in a Bayesian setting, and in Section~\ref{sec:minimax}, we discuss the problem in the minimax setting of Lorden (\cite{lord-amstat-1971}). We provide a brief review of the existing literature and our contributions as compared to it in Section~\ref{sec:LitReviewContributions}.

The general solution approach taken is as follows. We assume that there are families of distributions $\{\mathcal{P}_{n}\}$ such that
$$
g_{n} \in \mathcal{P}_{n}, \quad n = 1,2, \dots, \quad n \geq \nu
$$
and the families $\{\mathcal{P}_{n}\}$ are known to the decision maker. We further assume that there exist densities $\bar{g}_{n} \in \mathcal{P}_{n}$ that the density sequence $ \{\bar{g}_{n}\}$ is least favorable in a well-defined sense (this notion will be made rigorous below). Then, under mild additional conditions, the generalized versions of the Shiryaev test and the CUSUM test designed using the least favorable densities $ \{\bar{g}_{n}\}$ are robust (exactly or asymptotically) optimal for their respective problem settings (Section~\ref{sec:Bayes} and Section~\ref{sec:minimax}). 
In Section~\ref{sec:exactrobustopt}, we provide conditions under which the CUSUM and the Shiryaev algorithms are exactly robust optimal for non-stationary processes. In Section~\ref{sec:AsymptoticOpt}, we provide a set of sufficient conditions for the CUSUM and the Shiryaev algorithms to be asymptotically robust optimal for non-stationary processes. In Section~\ref{sec:LFDexamples}, we provide examples of \rcn{least favorable laws} from Gaussian and Poisson families of distributions. 
In Section~\ref{sec:NumericalResults}, we provide several numerical results to show the effectiveness of the robust tests. Specifically, we apply a robust test to detect arriving or approaching aircraft using aviation data. We also apply a robust test to detect the onset of a pandemic using COVID-19 daily infection data. Finally, we compare, using simulations, the performance of a robust test with a test that is not necessarily designed to be robust. 

\subsection{Motivating Applications}
\label{sec:motivation}
 \begin{figure}
    \centering\includegraphics[scale=0.45]{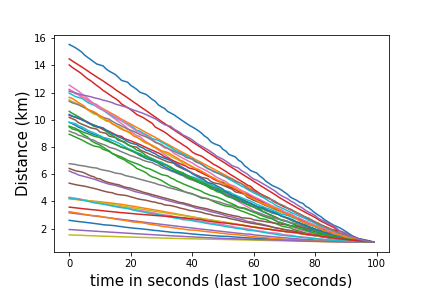}
     \includegraphics[scale=0.45]{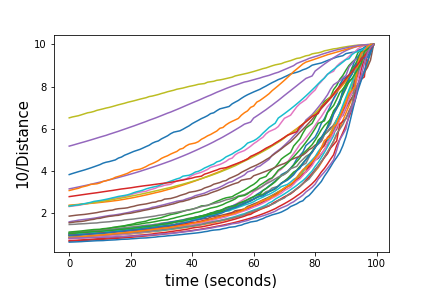}
     \caption{Distance measurements and corresponding signals extracted from datasets on aircraft trajectories collected from aircraft around the Pittsburgh-Butler Regional Airport (\cite{Patrikar2021}). }
     \label{fig:exampleFlights}
 \end{figure}

The problem of quickest change detection with non-stationary post-change distribution is encountered in the following applications: 
\begin{enumerate}
\item \textit{Satellite safety}: As discussed in \cite{brucks2023modeling}, one of the major challenges in satellite safety is to detect approaching debris that can cause damage to the satellite. As debris approaches the satellite, the corresponding measurements are expected to grow stochastically with time leading to a non-stationary post-change process. 
    \item \textit{Military applications}: A classical problem in military applications is the problem of detection of an arriving enemy aircraft or an enemy object (\cite{brucks2023modeling}). Similar to the satellite safety example, an approaching aircraft will lead to a stochastically growing process. As an example, in Fig.~\ref{fig:exampleFlights}, we have plotted distance and signal measurements extracted from aircraft trajectory datasets collected around Pittsburgh-Butler Regional Airport (\cite{Patrikar2021}). 
    \item \textit{Public health application}: In the post-COVID pandemic era, it is of significant interest to detect the onset of a pandemic. In Fig.~\ref{fig:exampleCOVID}, we have plotted the daily infection numbers for Allegheny and St. Louis counties for the first $200$ days starting $2020/1/22$. As seen in the figure, the numbers grow and then subside over time.  
    \item \textit{Social networking}: Another possible application is trend detection in social network data. The trending topic will cause a sudden increase in the number of messages, and these numbers can fluctuate over time. 
\end{enumerate}
In all the above applications, the exact manner in which the change will occur is generally not known to the decision maker, and they may have to take a conservative approach to the design of the detection system. A robust solution can guarantee that all possible post-change scenarios can be detected. 

 \begin{figure}[h]
    \centering\includegraphics[scale=0.45]{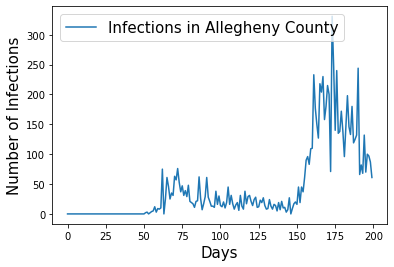}
     \includegraphics[scale=0.45]{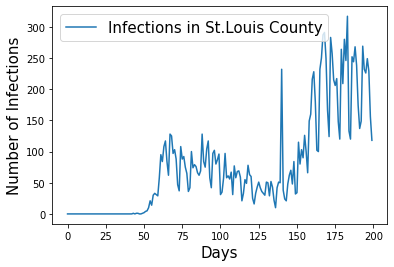}
     \caption{Daily infection rates for Allegheny and St. Louis counties in the first $200$ days starting 2020/1/22. The number of infections over time (even beyond the dates shown here) has multiple cycles of high values and low values. }
     \label{fig:exampleCOVID}
 \end{figure}

\vspace{-0.5cm}

\subsection{Contributions and Comparison with Existing Literature}
\label{sec:LitReviewContributions}
In this paper, we develop robust optimal algorithms for QCD when the post-change law is non-stationary and unknown. 

In the QCD literature, optimal algorithms are developed when the post-change distribution is unknown under three categories: 1) generalized likelihood ratio (GLR) tests are developed where the unknown post-change parameter is replaced by its maximum likelihood estimate (\cite{lord-amstat-1971, lai-ieeetit-1998, tart-book-2019, tart-niki-bass-2014}), 2) mixture-based tests are developed where a prior model is assumed on the post-change parameters and the likelihood ratio is integrated over this prior (\cite{lai-ieeetit-1998, poll-astat-1987, tart-book-2019, tart-niki-bass-2014}), and 3) robust tests are developed where the optimal tests are designed using the \rcn{least favorable law} (\cite{unni-etal-ieeeit-2011, oleyaeimotlagh2023quickest}). Of these three approaches, only the robust approach leads to statistics that can be calculated recursively. 

Optimal algorithms for non-stationary post-change models are developed in \cite{lai-ieeetit-1998, tart-book-2019, tart-niki-bass-2014, tart-veer-siamtpa-2005} under non-i.i.d. setting, and in \cite{brucks2023modeling, liang2022quickest} in independent setting. While GLR and mixture approaches have been taken in these works, a result with a robust approach is not available in the literature. 

A salient feature of our work is that this is the first work where \textit{exactly} optimal tests are developed in non-stationary and unknown distribution settings (see Section~\ref{sec:exactrobustopt}). Exactly optimal solutions for unknown post-change distributions are developed in \cite{unni-etal-ieeeit-2011}. However, there the post-change model is i.i.d. (hence stationary). Exactly optimal algorithms are developed in a non-stationary setting in \cite{bane-tit-2021}. However, the post-change model is assumed known. 
Both non-stationary and unknown distributions are considered in \cite{oleyaeimotlagh2023quickest}. But, the non-stationary model is limited to a statistically periodic model. In this paper, we consider a much broader (almost arbitrary) class of non-stationary processes. Exact robust optimality for non-stationary processes is also considered in \cite{molloy2018minimax}. However, the notion of the least favorable law (LFL) is stationary with time. In this paper, we allow the LFL to also change with time. This allows for the study of robustness when an exactly optimal solution is obtained in a non-stationary setting; see \cite{bane-tit-2021}. 

We also develop asymptotic optimality theory for non-stationary and unknown post-change models. 
For both exact and asymptotic optimality theory, we consider stochastic boundedness assumptions similar to \cite{unni-etal-ieeeit-2011} and \cite{oleyaeimotlagh2023quickest}. In this regard, we note that asymptotic optimality theories have also been developed under weaker stochastic boundedness assumptions in \cite{molloy2017misspecified} and \cite{liang2022non}. Also, see \cite{xie2023distributionally} for another approach. Since we use a unified framework to develop both exact and asymptotic optimality theory, and especially because of our interest in exact optimality, we consider the stronger stochastic boundedness assumptions of \cite{unni-etal-ieeeit-2011} and \cite{oleyaeimotlagh2023quickest}.

\section{Robust Optimality in a Bayesian Setting}\label{sec:Bayes}

Our change point model is as specified in \eqref{eq:changepointmodel_0}. Specifically, we assume that we observe a sequence of independent random variables $\{X_n\}$ over time. Before a point $\nu$, called a change point in the following, the process is i.i.d. with density $f$. After time $\nu$, the law of the process changes to a sequence of densities $\{g_{n}\}$. Mathematically, 
\begin{equation}\label{eq:changepointmodel}
	X_n \sim
	\begin{cases}
		f, &\quad \forall n < \nu, \\
		g_{n}, &\quad \forall n \geq \nu.
	\end{cases}
\end{equation}
Thus, the post-change distribution may depend on the observation time $n$. 
The post-change densities are assumed to satisfy
\begin{equation*}
	D(g_{n} \; \| \; f) > 0, \quad \forall n \geq 1, \;   \nu \geq 1, \; n \geq \nu, 
\end{equation*}
where $D(g \; \| \; f) = \int g(x) \log [g(x)/f(x)] dx$ is the Kullback-Leibler divergence between densities $g$ and $f$. As mentioned in the introduction, with a slight abuse of notation (borrowing it from the time-series literature (\cite{shumway2010time})), we will refer to a process with densities $\{g_{n}\}$ as a non-stationary process since the density of the data evolves or changes with time.
Below, we use the notation $G = \{g_{n}\}$ to denote the post-change sequence of densities.

For the problem of robust quickest change detection in non-stationary processes, we assume that the post-change law $G$  is unknown. However, there are families of distributions $\{\mathcal{P}_{n}\}$ such that
$$
g_{n} \in \mathcal{P}_{n}, \quad n, \nu = 1,2, \dots,
$$
and the families $\{\mathcal{P}_{n}\}$ are known to the decision maker. Below, we use the notation
$$
\mathcal{G} = \{G = \{g_{n}\}: g_{n} \in \mathcal{P}_{n}, \; n, \nu \geq 1\}
$$
to denote the set of all possible post-change laws.

Let $\tau$ be a stopping time for the process $\{X_n\}$, i.e., a positive integer-valued random variable such that the event $\{\tau \leq n\}$ belongs
to the $\sigma$-algebra generated by $X_1, \cdots, X_n$. In other words, whether or not $\tau \leq n$ is completely determined by the first $n$
observations. We declare that a change has occurred at the stopping time $\tau$. To find the best stopping time to detect the change in distribution, we
need a performance criterion. Towards this end, we model the change point $\nu$ as a random variable with a prior distribution given by
$$
\pi_n = \Prob(\nu =n), \quad n = 1, 2, \cdots.
$$
For each $n \in \mathbb{N}$, we use $\Prob_n^{G}$ to denote the law of the observation process $\{X_n\}$ when the change occurs at $\nu=n$ and the post-change law is given by $G$. We use $\Expect_n^{G}$ to denote the corresponding expectation. The notations $\Prob_\infty^{G}=\Prob_\infty$ and $\Expect_\infty^{G}=\Expect_\infty$ are used when there is no change (suggesting a change occurring at $\nu=\infty$). Note that in case of no change, the superscript $G$ can be dropped from the notation as the post-change law plays no role. 
Using these notations, we define the average probability measure
$$
\Prob^{\pi,G} = \sum_{n=1}^\infty \pi_n \; \Prob_n^G.
$$
To capture a penalty for the false alarms,
in the event that the stopping time occurs before the change,
we use the probability of a false alarm defined as
$$
\Prob^{\pi,G}(\tau < \nu).
$$
Note that the probability of a false alarm
$\Prob^{\pi,G}(\tau < \nu)$ is not a function of the post-change law $G$. This is because
$$
    \Prob^{\pi, G}(\tau < \nu) = \sum_{n=1}^\infty \pi_n \Prob_n^G(\tau < n) = \sum_{n=1}^\infty \pi_n \Prob_\infty(\tau < n), 
    $$
    where the last equality follows because $\{\tau < n\}$ is a function of $X_1, \dots, X_{n-1}$, and the law of $X_1, \dots, X_{n-1}$ is the same under both $\Prob_n^G$ and $\Prob_\infty$. 
Hence, in the following, we suppress the mention of $G$ and refer to the probability of false alarm only by
$$
\Prob^{\pi}(\tau < \nu).
$$
To penalize the detection delay, we use the average detection delay given by
$$
\Expect^{\pi,G}\left[(\tau - \nu)^+\right],
$$
where $x^+ = \max\{x, 0\}$. 
\rcn{The classical QCD formulation in this Bayesian setting is given by (\cite{shir-siamtpa-1963})
\begin{equation}\label{eq:QCDproblem1}
    \inf_{\tau \in \textup{C}_\alpha} 
    \Expect^{\pi} \left[(\tau - \nu )^+ \right],
\end{equation}
where the constraint set is the following
$$
\textup{C}_\alpha = \left\{\tau:  \Prob^{\pi}(\tau < \nu) \leq \alpha\right\},
$$
and $\alpha$ is a given constraint on the probability of a false alarm.
However, since we assume that the decision maker is not aware of post-change law $G$, the optimization problem we are interested in solving is}
\begin{equation}\label{eq:robustProb}
	\inf_{\tau \in \textup{C}_\alpha} \;\; \sup_{G \in \mathcal{G}} \; \Expect^{\pi,G}\left[(\tau - \nu)^+\right].
\end{equation}
We say that a solution is robust optimal (in the Bayesian setting) if a stopping time or algorithm is a solution to the problem in \eqref{eq:robustProb}.

We now provide the optimal or asymptotically optimal solution to \eqref{eq:robustProb} under assumptions on the families of post-change uncertainty classes $\{\mathcal{P}_{n}\}$. Specifically, we extend the results in \cite{unni-etal-ieeeit-2011} for i.i.d. processes and those in \cite{oleyaeimotlagh2023quickest} for independent and periodically identically distributed (i.p.i.d.) processes 
to non-stationary processes.
We assume in the rest of this section that all densities involved are equivalent to each other (absolutely continuous with respect to each other). 

To state the assumptions on $\{\mathcal{P}_{n}\}$, we need some definitions. We say that a random variable $Z_2$ is stochastically larger than another random variable $Z_1$ if
$$
\Prob(Z_2 \geq t) \geq \Prob(Z_1 \geq t), \quad \forall t \in \mathbb{R}.
$$
We use the notation
$$
Z_2 \succ Z_1.
$$
If $\mathcal{L}_{Z_2}$ and $\mathcal{L}_{Z_1}$ are the probability laws of $Z_2$ and $Z_1$, then we also use the notation
$$
\mathcal{L}_{Z_2} \succ \mathcal{L}_{Z_1}.
$$
We now introduce the notion of stochastic boundedness in non-stationary processes. In the following, we use
$$
\mathcal{L}(\phi(X), g)
$$
to denote the law of some function $\phi(X)$ of the random variable $X$, when the variable $X$ has density $g$.

\medspace
\medspace
\medspace
\medspace
\medspace

\begin{definition}[Stochastic Boundedness in Non-Stationary Processes; Least Favorable Law (LFL)]
	We say that the family $\{\mathcal{P}_{n}\}$ is stochastically bounded by the sequence
	$$
	\bar{G}=\{\bar{g}_{n}\},
	$$
	and call $\bar{G}$ the least favorable law (LFL), if
	$$
	\bar{g}_{n} \in \mathcal{P}_{n}, \quad n=1,2, \dots,
	$$
	and
	\begin{equation}
 \label{eq:stocbounded}
		\begin{split}
			\mathcal{L}\left(\log \frac{\bar{g}_{n}(X_n)}{f(X_n)}, g_{n}\right) &\succ 	\mathcal{L}\left(\log \frac{\bar{g}_{n}(X_n)}{f(X_n)},
			\bar{g}_{n}\right), \quad \forall  g_{n} \in \mathcal{P}_{n}, \quad n, \nu \geq 1, \quad n \geq \nu
		\end{split}
	\end{equation}
\end{definition}

Let
\begin{align}
    \label{eqn:geneShir}
     \bar{R}_n = \frac{1}{\Prob(\nu > n)}\sum_{k=1}^n \pi_k \; \prod_{i=k}^n \frac{\bar{g}_{i}(X_i)}{f(X_i)} 
\end{align}
be the generalized Shiryaev statistic for the change point model \eqref{eq:changepointmodel} and for the LFL $\bar{G}$. If the change point is a geometrically distributed random variable with parameter $\rho$, then the statistic becomes
\begin{align}
    \label{eqn: Rn}
     \bar{R}_n = \frac{1}{(1-\rho)^n}\sum_{k=1}^n (1-\rho)^{k-1}\rho \; \prod_{i=k}^n \frac{\bar{g}_{i}(X_i)}{f(X_i)}.
\end{align}
Also let $\bar{\tau}^*$ be the generalized Shiryaev stopping rule
\begin{equation}\label{eq:LFLshir}
	\bar{\tau}^* = \inf \{n \geq 1: \bar{R}_n \geq \bar{A}_{n,\alpha} \}. 
\end{equation}
Here thresholds $\bar{A}_{n,\alpha}$ are chosen to satisfy the constraint on the probability of a false alarm. 

We note that additional assumptions are needed on the densities $\{\bar{g}_{n}\}$ and the thresholds $\{\bar{A}_{n,\alpha}\}$ to guarantee optimality or asymptotic optimality of the Shiryaev stopping rule \eqref{eq:LFLshir}. Of special interest is the scenario when
$$
\bar{g}_{n} = \bar{g}, \quad \forall n, \nu \geq 1, \quad n \geq \nu 
$$
In this case, the LFL corresponds to an i.i.d. process with the law $\bar{g}$. Furthermore, if the change point is geometric and the threshold $\bar{A}_{n,\alpha} = \bar{A}_\alpha$ is selected such that the probability of a false alarm is exactly equal to $\alpha$, then the stopping rule $\bar{\tau}^*$ is exactly optimal when $G=\bar{G}$. In addition to this, the statistic $\bar{R}_n$ also has a recursive implementation. More conditions on exact and asymptotic optimality are discussed in Sections~\ref{sec:exactrobustopt} and Section~\ref{sec:AsymptoticOpt}. For some general conditions, see \cite{tart-veer-siamtpa-2005, tart-niki-bass-2014, tart-book-2019, bane-tit-2021}.  

The main result of this section is Theorem~\ref{thm:LFLRobust} below where we show that if the generalized Shiryaev algorithm is optimal for \rcn{\eqref{eq:QCDproblem1} with $G$ being} the LFL, then it is robust optimal for the problem in \eqref{eq:robustProb}. For the theorem, we need the following lemma from \cite{unni-etal-ieeeit-2011} which we reproduce here for readability.
\begin{lemma}[\cite{unni-etal-ieeeit-2011}]
\label{lem:stocbound_UV}
    Suppose $\{U_i: 1 \leq i \leq n\}$ is a set of mutually independent random variables, and $\{V_i: 1 \leq i \leq n\}$ is another set of mutually independent random variables such that $U_i \succ V_i$, $1 \leq i \leq n$. Now let $q: \mathbb{R}^n \to \mathbb{R}$ be a continuous real-valued function defined on $\mathbb{R}^n$ that satisfies
    \begin{align*}
        q(X_1, \dots, X_{i-1}, a, X_{i+1}, \dots, X_{n}) \geq q(X_1, \dots, X_{i-1}, X_i, X_{i+1}, \dots, X_{n})
    \end{align*}
    for all $(X_1, \dots, X_n) \in \mathbb{R}^n$, $a > X_i$, and $i \in \{1, \dots, n\}$. Then we have 
    \begin{align*}
        q(U_1, U_2, \dots, U_n) \succ q(V_1, V_2, \dots, V_n).
    \end{align*}
\end{lemma}


\begin{theorem}
\label{thm:LFLRobust}
	Suppose the following conditions hold:
	\begin{enumerate}
		\item[(a)] 	The family $\{\mathcal{P}_{n}\}$ is stochastically bounded by the  law
		$
		\bar{G}= \{\bar{g}_{n}\}.
		$
  Thus, $\bar{G}=\{\bar{g}_{n}\}$ is the LFL. 
		\item[(b)] Let $\alpha \in (0,1)$ be a constraint such that
		$
		\Prob^\pi(\bar{\tau}^* < \nu) = \alpha,
		$
		where $\bar{\tau}^*$ is the optimal rule designed using the LFL \eqref{eq:LFLshir}.
		\item[(c)] All likelihood ratio functions involved are continuous.
  \end{enumerate}
  Then the following results are true:
  \begin{enumerate}
    \item \rcn{If the stopping rule $\bar{\tau}^*$ in \eqref{eq:LFLshir} is exactly optimal for the problem in \eqref{eq:QCDproblem1} when the post-change law is LFL, namely $\{\bar{G}\}$}, then the stopping rule is exactly robust optimal for the problem in \eqref{eq:robustProb}.
    \item \rcn{If the stopping rule $\bar{\tau}^*$ in \eqref{eq:LFLshir} is asymptotically optimal for the problem in \eqref{eq:QCDproblem1} when the post-change law is LFL, namely $\{\bar{G}\}$}, then the stopping rule is asymptotically robust optimal for the problem in \eqref{eq:robustProb}.
	\end{enumerate}
	
\end{theorem}

\begin{proof}
	The key step in the proof is to show that for each $k \in \mathbb{N}$, $G \in \mathcal{G}$,
 \begin{equation}\label{eq:keystep}
		\begin{split}
			\Expect_k^{\bar{G}}&\left[(\bar{\tau}^* - k)^+ | \mathcal{F}_{k-1}\right] \geq \Expect_k^{G}\left[(\bar{\tau}^* - k)^+ | \mathcal{F}_{k-1}\right], 
		\end{split}
	\end{equation}
 where $\mathcal{F}_{k-1}$ is the $\sigma$-algebra generated by $X_1, \dots, X_{k-1}$. 
	If the above statement is true, then by taking expectation 
 and then taking an average over the prior on the change point, we get
	\begin{equation}
		\begin{split}
			\Expect^{\pi,\bar{G}}\left[(\bar{\tau}^* - \nu)^+\right] &=  \sum_k \pi_k \Expect_k^{\bar{G}}\left[(\bar{\tau}^* - k)^+\right] \\
			&\geq \sum_k \pi_k \Expect_k^{G}\left[(\bar{\tau}^* - k)^+ \right] =	\Expect^{\pi,G}\left[(\bar{\tau}^* - \nu)^+ \right], \quad \forall G \in \mathcal{G}.
		\end{split}
	\end{equation}
	The last equation gives
	\begin{equation}
		\begin{split}
			\Expect^{\pi,\bar{G}}&\left[(\bar{\tau}^* - \nu)^+\right]  \geq \Expect^{\pi,G}\left[(\bar{\tau}^* - \nu)^+ \right], \quad \forall G \in \mathcal{G}.
		\end{split}
	\end{equation}
	This implies that
	\begin{equation}
		\begin{split}
			\Expect^{\pi,\bar{G}}&\left[(\bar{\tau}^* - \nu)^+\right]  = \sup_{G \in \mathcal{G}} \Expect^{\pi,G}\left[(\bar{\tau}^* - \nu)^+ \right],
		\end{split}
	\end{equation}
 where we have equality because the law $\bar{G}$ belongs to the family considered on the right. 
	Now, if $\tau$ is any stopping rule satisfying the probability of false alarm constraint of $\alpha$, then since $\bar{\tau}^*$ is the optimal test for the LFL $\bar{G}$, we have 
	\begin{equation}
		\begin{split}
			\sup_{G \in \mathcal{G}} \Expect^{\pi,G}\left[(\tau - \nu)^+ \right] &\geq \Expect^{\pi,\bar{G}}\left[(\tau - \nu)^+\right]
			\geq
			\Expect^{\pi,\bar{G}}\left[(\bar{\tau}^* - \nu)^+\right]  (1+o^*(1))\\
			&= \sup_{G \in \mathcal{G}} \Expect^{\pi,G}\left[(\bar{\tau}^* - \nu)^+ \right] (1+o^*(1)).
		\end{split}
	\end{equation}
        Here the term $o^*(1)$ is defined as 
        \begin{equation}
            o^*(1) = 
            \begin{cases}
                0, &\quad \text{if $\bar{\tau}^*$ is exactly optimal,}\\
                o(1), &\quad \text{if $\bar{\tau}^*$ is asymptotically optimal}.\\
            \end{cases}
        \end{equation}
         Thus, if the rule $\bar{\tau}^*$ is exactly optimal \rcn{for \eqref{eq:QCDproblem1} with post-change $\bar{G}$}, then the $o^*(1)$ terms will be identically zero, and we will get exact robust optimality of $\bar{\tau}^*$ for the problem in \eqref{eq:robustProb}. If the rule $\bar{\tau}^*$ is asymptotically optimal \rcn{for \eqref{eq:QCDproblem1} with post-change $\bar{G}$}, then the $o^*(1)$ terms is the same as $o(1)$ which goes to zero in the limit $\alpha \to 0$. 
	
We now prove the key step \eqref{eq:keystep}. Towards this end, we prove that for every integer $N\geq 0$ and $k \geq 1$, 
\begin{equation}\label{eq:keystep2}
		\begin{split}
			\Prob_k^{\bar{G}}\left[(\bar{\tau}^* - k)^+ > N |\mathcal{F}_{k-1}\right] &\geq \Prob_k^{G}\left[(\bar{\tau}^* - k)^+ > N |\mathcal{F}_{k-1}\right], \quad  \forall G \in \mathcal{G}.
		\end{split}
	\end{equation}
        Equation \eqref{eq:keystep2} implies \eqref{eq:keystep} because through \eqref{eq:keystep2} we would establish that $(\bar{\tau}^* - k)^+$ is stochastically bigger under the conditional distribution when the law is $\bar{G}$ and the conditioning is on $\mathcal{F}_{k-1}$. In fact, a summation over $N$ in \eqref{eq:keystep2} would imply \eqref{eq:keystep}. 
        
To prove \eqref{eq:keystep2}, we show that for $N \geq 0$ and $k \geq 1$, 
 \begin{equation}\label{eq:keystep3}
		\begin{split}
			\Prob_k^{\bar{G}}\left[(\bar{\tau}^* - k)^+ \leq N |\mathcal{F}_{k-1}\right] &\leq \Prob_k^{G}\left[(\bar{\tau}^* - k)^+ \leq N |\mathcal{F}_{k-1}\right], \quad  \forall G \in \mathcal{G}.
		\end{split}
	\end{equation}
First note that the events $\{(\bar{\tau}^* - k)^+ \leq N\}$ and $\{(\bar{\tau}^* - k) \leq N\}$ are identical:
\begin{equation}
\label{eq:eventequiv_1}
    \begin{split}
       \{(\bar{\tau}^* - k)^+ \leq N\} &= \left[ \{(\bar{\tau}^* - k)^+ \leq N\} \cap \{\bar{\tau}^* \geq k \}\right] \cup \left[ \{(\bar{\tau}^* - k)^+ \leq N\} \cap \{\bar{\tau}^* < k \}\right]\\
       &= \left[ \{(\bar{\tau}^* - k) \leq N\} \cap \{\bar{\tau}^* \geq k \}\right] \cup \left[ \{(\bar{\tau}^* - k)^+ \leq N\} \cap \{\bar{\tau}^* < k \}\right]\\
       &\stackrel{(a)}{=} \left[ \{(\bar{\tau}^* - k) \leq N\} \cap \{\bar{\tau}^* \geq k \}\right] \cup \left[  \{\bar{\tau}^* < k \}\right]\\
       &\stackrel{(b)}{=} \left[ \{(\bar{\tau}^* - k) \leq N\} \cap \{\bar{\tau}^* \geq k \}\right] \cup \left[ \{(\bar{\tau}^* - k) \leq N\} \cap \{\bar{\tau}^* < k \}\right]\\
       &=\{(\bar{\tau}^* - k) \leq N\}.
    \end{split}
\end{equation}
Here the equalities $(a)$ and $(b)$ follow because for $N \geq 0$,
\begin{equation*}
    \begin{split}
        \{\bar{\tau}^* < k \}  &\subset \{(\bar{\tau}^* - k)^+ \leq N\}, \\
        \{\bar{\tau}^* < k \}  &\subset \{(\bar{\tau}^* - k) \leq N\}.
    \end{split}
    \end{equation*}
This equivalence of events implies that
	\begin{equation}\label{eq:temp1}
		\begin{split}
			\Prob_k^{\bar{G}}\left[(\bar{\tau}^* - k)^+ \leq N |\mathcal{F}_{k-1}\right] 
			&= \Prob_k^{\bar{G}}\left[\bar{\tau}^* \leq k+N |\mathcal{F}_{k-1}\right] \\
			& = \Prob_k^{\bar{G}}\left[l(X_1, X_2, \dots, X_{k+N}) \; \geq \; 0 \; | \; \mathcal{F}_{k-1}\right],
		\end{split}
	\end{equation}
where the function $l(z_1, z_2, \dots, z_{N})$ is given by
	\begin{equation}
		\begin{split}
			l(z_1, z_2, \dots, z_N) 
			&= \max_{1 \leq n \leq N} \left(\frac{1}{\Prob(\nu > n)}\sum_{k=1}^n \pi_k \; \exp \left(\sum_{i=k}^n \log [\bar{g}_{i}(z_{i})/f(z_i)]\right)-  \bar A_{n, \alpha}\right).
		\end{split}
	\end{equation}
Thus, $l(X_1, X_2, \dots, X_{k+N}) \geq 0$ in (\ref{eq:temp1}) means that the Shiryaev statistics defined in \eqref{eqn:geneShir} exceeds the sequence of thresholds $\{ \bar A_{n, \alpha}\}$ before time $k + N$.
Using the stochastic boundedness assumption, for all $g_{n} \in \mathcal{P}_{n}, n, \nu=1,2, \dots, n \geq \nu$, 
 \begin{equation}\label{eq:stocboundcond_again}
		\begin{split}
			\mathcal{L}\left(\log \frac{\bar{g}_{n}(X_n)}{f(X_n)}, g_{n}\right) &\succ 	\mathcal{L}\left(\log \frac{\bar{g}_{n}(X_n)}{f(X_n)},
			\bar{g}_{n}\right),
		\end{split}
	\end{equation}
 and the fact that $l(z_1, z_2, \dots, z_{N})$ is continuous and increasing in $\log [\bar{g}_{i}(z_{i})/f(z_i)]$, we have by Lemma~\ref{lem:stocbound_UV} that  
	\begin{equation}\label{eq:temp3}
		\begin{split}
			\Prob_k^{\bar{G}}\left[(\bar{\tau}^* - k)^+ \leq N \; | \; \mathcal{F}_{k-1}\right] 
			 &= \Prob_k^{\bar{G}}\left[\bar{\tau}^* \leq k+N \; |\; \mathcal{F}_{k-1}\right] \\
			& = \Prob_k^{\bar{G}}\left[l(X_1, X_2, \dots, X_{k+N}) \geq 0 \; |\; \mathcal{F}_{k-1}\right] \\
			&\stackrel{(c)}{\leq} \Prob_k^{{G}}\left[l(X_1, X_2, \dots, X_{k+N}) \geq 0 \; |\; \mathcal{F}_{k-1}\right] \\ 
   &= \Prob_k^{G}\left[\bar{\tau}^* \leq k+N \; |\; \mathcal{F}_{k-1}\right] \\
			&=\Prob_k^{{G}}\left[(\bar{\tau}^* - k)^+ \leq N \; |\; \mathcal{F}_{k-1}\right], \quad  \forall G \in \mathcal{G}.
		\end{split}
	\end{equation}
In the above equation, the inequality $(c)$ is true because the conditioning on the change point and the past realization $\mathcal{F}_{k-1}$ fixes the first $i = 1, \dots, k-1$ coordinates of $\{\log [\bar{g}_{i}(z_{i})/f(z_i)]\}$. Also, the rest of the coordinates of $\{\log [\bar{g}_{i}(z_{i})/f(z_i)]\}$ are independent, with their law being $\bar{G}$ to the left of inequality, and $G$ to the right.  We can now invoke the stochastic boundedness condition \eqref{eq:stocboundcond_again}. 
	This proves \eqref{eq:keystep2} and hence \eqref{eq:keystep}.
\end{proof}

\begin{remark}
   We note that to ensure the exact optimality of the generalized Shiryaev algorithm, we may have to assume that the change point is a geometrically distributed random variable. 
\end{remark}
We discuss special cases of this result and examples in Section~\ref{sec:exactrobustopt} to Section~\ref{sec:LFDexamples}.


\section{Robust Optimality in a Minimax Setting}
\label{sec:minimax}
In this section, we obtain a solution for a robust version of the minimax problem of Lorden (\cite{lord-amstat-1971}). We have the same change point model as in Section~\ref{sec:Bayes}. But, here we assume that the change point $\nu$ is an unknown constant.
\rcn{The standard Lorden formulation of the quickest change detection is:
\begin{equation}\label{eq:QCDproblem2}
    \inf_{\tau \in \textup{D}_\alpha} \;\; 
    \sup_\nu \; \esssup \; \Expect_\nu \left[(\tau - \nu +1)^+ | \mathcal{F}_{\nu-1}\right],
\end{equation}
where $\esssup X$ represents the smallest constant $C$ such that $\Prob(X \leq C)=1$, and 
\begin{equation*}
\textup{D}_\alpha = \left\{\tau: \Expect_\infty[\tau] \geq \frac{1}{\alpha}, \right\},
\end{equation*}
with $1/\alpha$ being a constraint on the meantime to a false alarm for $\alpha \in (0,1)$.
However, we assume that the decision maker is not aware of post-change law $G$. Thus, the optimization problem we are interested in solving is  }
\begin{equation}\label{eq:robustProbmini}
	\inf_{\tau \in \textup{D}_\alpha} \;\; \sup_{G \in \mathcal{G}} \; \; \sup_\nu \; \esssup \; \Expect_\nu^{G}\left[(\tau - \nu + 1)^+ | \mathcal{F}_{\nu-1}\right].
\end{equation}
\rcn{We say that a solution is robust optimal (in the minimax setting) if a stopping time or algorithm is a solution to the problem in \eqref{eq:robustProbmini}}. 

\rc{In the rest of the paper, we use the notation
\begin{equation}
\label{eqn:cusum stat}
    \text{WADD}^G(\tau) \coloneqq \sup_\nu \; \esssup \; \Expect_\nu^{G}\left[(\tau - \nu + 1)^+ | \mathcal{F}_{\nu-1}\right].
\end{equation}}

Let
$$
 \rc{\bar{W}_n} = \max_{1 \leq k \leq n} \sum_{i=k}^n \log \frac{\bar{g}_{i}(X_i)}{f(X_i)} 
$$
be the generalized CUSUM statistic for the LFL $\bar{G}$. 
Also let $\bar{\tau}_c^*$ be the stopping rule
\begin{equation}\label{eq:LFLCUSUM}
	\bar{\tau_c}^* = \inf \{n \geq 1: \rc{\bar{W}_n} \geq \rc{\bar B_{n,\alpha}} \}. 
\end{equation}
Here thresholds \rc{$\bar B_{n,\alpha}$} are chosen to satisfy the constraint on the mean time to a false alarm. We again note that additional assumptions are needed on the densities $\{\bar{g}_{n}\}$ to guarantee optimality or asymptotic optimality of the generalized CUSUM stopping rule \eqref{eq:LFLCUSUM}. For some general conditions, see \cite{lai-ieeetit-1998, tart-niki-bass-2014, tart-book-2019, liang2022quickest}.  

We now state \rc{and prove} our main result on robust minimax quickest change detection in non-stationary processes. 


\begin{theorem}
\label{thm:LFLRobust_minimax}
	\rc{Suppose the following conditions hold:
	\begin{enumerate}
		\item[(a)] 	The family $\{\mathcal{P}_{n}\}$ is stochastically bounded by the  law
		$
		\bar{G}= \{\bar{g}_{n}\}.
		$
  Thus, $\bar{G}=\{\bar{g}_{n}\}$ is the LFL. 
		\item[(b)] Let $\alpha \in (0,1)$ be a constraint such that
		$
		\Expect_\infty[\bar{\tau}_c^*] = \frac{1}{\alpha},
		$
		where $\bar{\tau}_c^*$ is the optimal rule designed using the LFL.
		\item[(c)] All likelihood ratio functions involved are continuous.
  \end{enumerate}}
\rc{Then the following results are true:
  \begin{enumerate}
		\item If the stopping rule $\bar{\tau}_c^*$ in \eqref{eq:LFLCUSUM} is exactly optimal \rcn{for \eqref{eq:QCDproblem2} when the post-change law is LFL, namely $\{\bar{G}\}$,} then the stopping rule is exactly robust optimal for the problem in \eqref{eq:robustProbmini}.
 \item If the stopping rule $\bar{\tau}_c^*$ in \eqref{eq:LFLCUSUM} is asymptotically optimal \rcn{for \eqref{eq:QCDproblem2} when the post-change law is LFL, namely $\{\bar{G}\}$,} then the stopping rule is asymptotically robust optimal for the problem in \eqref{eq:robustProbmini}.
	\end{enumerate}}
	
\end{theorem}

\begin{proof}
\rc{The proof is similar to the proof of Theorem~\ref{thm:LFLRobust}. We provide proof for completeness. }

The key step in the proof is to show that for each $k \in \mathbb{N}$, $G \in \mathcal{G}$,
 \begin{equation}\label{eq:keystep3_1}
		\begin{split}
			\Expect_k^{\bar{G}}&\left[(\bar{\tau}_c^* - k + 1)^+ | \mathcal{F}_{k-1}\right] \geq \Expect_k^{G}\left[(\bar{\tau}_c^* - k+ 1)^+ | \mathcal{F}_{k-1}\right]. 
		\end{split}
	\end{equation}
 To prove the above statement, we need to show that for every integer $N\geq 0$,
	\begin{equation}\label{eq:keystep4}
		\begin{split}
			\Prob_k^{\bar{G}}\left[(\bar{\tau}_c^* - k+ 1)^+ > N |\mathcal{F}_{k-1}\right] &\geq \Prob_k^{G}\left[(\bar{\tau}_c^* - k+ 1)^+ > N |\mathcal{F}_{k-1}\right],  \quad \forall G \in \mathcal{G}.
		\end{split}
	\end{equation}
	\rc{We prove this inequality for $N=0$ and $N \geq 1$ separately. We note that this approach was not taken in the proof of Theorem~\ref{thm:LFLRobust} because the delay was penalized as $(\tau-k)^+$. For $N=0$, note that 
 $$
 \{(\bar{\tau}_c^* - k + 1)^+ > 0\} = \{\bar{\tau}_c^* - k + 1 > 0\} = \{\bar{\tau}_c^* >  k- 1\} = \{\bar{\tau}_c^* \leq  k- 1\}^c. 
 $$
 Since the event $\{\bar{\tau}_c^* \leq  k- 1\}$ is $\mathcal{F}_{k-1}$ measurable, both sides of the inequality \eqref{eq:keystep4} are equal to the indicator of the event $\{(\bar{\tau}_c^* - k + 1)^+ > N\}$. }
 
We now prove \eqref{eq:keystep4} for $N \geq 1$. 
 \rc{We first have (see \eqref{eq:eventequiv_1})}
	\begin{equation}\label{eq:temp1-1}
		\begin{split}
			\Prob_k^{\bar{G}}\left[(\bar{\tau}_c^* - k {+ 1})^+ \leq N |\mathcal{F}_{k-1}\right] 
			&= \Prob_k^{\bar{G}}\left[\bar{\tau}_c^* \leq k+N {-1} |\mathcal{F}_{k-1}\right] \\
			& = \Prob_k^{\bar{G}}\left[h(X_1, X_2, \dots, X_{k+N{-1}}) \; \geq \; 0 \; | \; \mathcal{F}_{k-1}\right],
		\end{split}
	\end{equation}
	where now
 the function $h(z_1, z_2, \dots, z_{N})$ is given by
	\begin{equation}
		\begin{split}
			h&(z_1, z_2, \dots, z_N) 
			= \max_{1 \leq n \leq N} \left(\max_{1 \leq k \leq n} \left(\sum_{i=k}^n \log \frac{\bar{g}_{i}(z_{i})}{f(z_i)}\right)- \rc{\bar B_{n, \alpha}}\right).
		\end{split}
	\end{equation}
\rc{Thus, $h(X_1, X_2, \dots, X_{k+N{-1}}) \geq 0$ in (\ref{eq:temp1-1}) means that the CUSUM statistics defined in (\ref{eqn:cusum stat}) exceeds the sequence of thresholds $\{ \bar B_{n, \alpha}\}$ before time $k + N{-1}$.}
 
 Again, using the stochastic boundedness assumption,
 	\begin{equation*}
		\begin{split}
			\mathcal{L}\left(\log \frac{\bar{g}_{n}(X_n)}{f(X_n)}, g_{n}\right) &\succ 	\mathcal{L}\left(\log \frac{\bar{g}_{n}(X_n)}{f(X_n)},
			\bar{g}_{n}\right), \quad \forall g_{n} \in \mathcal{P}_{n}, \quad n=1,2, \dots, \quad n \geq \nu
		\end{split}
	\end{equation*}
 and the fact that even $h(z_1, z_2, \dots, z_{N})$ is continuous and increasing in $\log [\bar{g}_{i}(z_{i})/f(z_i)]$, we have by Lemma~\ref{lem:stocbound_UV} that (see comments regarding inequality $(c)$ in \eqref{eq:temp3})
	\begin{equation}\label{eq:temp4}
		\begin{split}
			\Prob_k^{\bar{G}}\left[(\bar{\tau}_c^* - k+1)^+ \leq N \; | \; \mathcal{F}_{k-1}\right] 
			 &= \Prob_k^{\bar{G}}\left[\bar{\tau}_c^* \leq k+N-1 \; |\; \mathcal{F}_{k-1}\right] \\
			& = \Prob_k^{\bar{G}}\left[h(X_1, X_2, \dots, X_{k+N-1}) \geq 0 \; |\; \mathcal{F}_{k-1}\right] \\
			&\leq \Prob_k^{{G}}\left[h(X_1, X_2, \dots, X_{k+N-1}) \geq 0 \; |\; \mathcal{F}_{k-1}\right] \\ 
   &= \Prob_k^{G}\left[\bar{\tau}_c^* \leq k+N-1 \; |\; \mathcal{F}_{k-1}\right] \\
			&=\Prob_k^{{G}}\left[(\bar{\tau}_c^* - k+1)^+ \leq N \; |\; \mathcal{F}_{k-1}\right], \quad  \forall G \in \mathcal{G}.
		\end{split}
	\end{equation}
This proves \eqref{eq:keystep4} and hence \rc{\eqref{eq:keystep3_1}}. Thus, we have
 \begin{equation*}
		\begin{split}
			\Expect_k^{\bar{G}}&\left[(\bar{\tau}_c^* - k{+1})^+ | \mathcal{F}_{k-1}\right] \geq \Expect_k^{G}\left[(\bar{\tau}_c^* - k{+1})^+ | \mathcal{F}_{k-1}\right], \quad \forall G \in \mathcal{G}.
		\end{split}
	\end{equation*}
This implies (by the definition of essential supremum) 
 \begin{equation}\label{eq:esssuporder}
		\begin{split}
			\esssup \; \Expect_k^{\bar{G}}&\left[(\bar{\tau}_c^* - k{+1})^+ | \mathcal{F}_{k-1}\right] \geq \esssup \; 
  \Expect_k^{G}\left[(\bar{\tau}_c^* - k{+1})^+ | \mathcal{F}_{k-1}\right], \quad \forall G \in \mathcal{G}.
		\end{split}
	\end{equation}
Since \eqref{eq:esssuporder} is true for every $k$, we have 
 \begin{equation}
		\begin{split}
		\sup_k\; 	\esssup \; \Expect_k^{\bar{G}}&\left[(\bar{\tau}_c^* - k{+1})^+ | \mathcal{F}_{k-1}\right] \geq \sup_k \; \esssup \; 
 \Expect_k^{G}\left[(\bar{\tau}_c^* - k{+1})^+ | \mathcal{F}_{k-1}\right], \;\; \forall G \in \mathcal{G}.
		\end{split}
	\end{equation}
This gives us
 \begin{equation}
		\sup_k\; 	\esssup \; \Expect_k^{\bar{G}}\left[(\bar{\tau}_c^* - k{+1})^+ | \mathcal{F}_{k-1}\right] = \sup_{G \in \mathcal{G}} \; \sup_k \; \esssup \; 
  \Expect_k^{G}\left[(\bar{\tau}_c^* - k{+1})^+ | \mathcal{F}_{k-1}\right],
 \end{equation}
or
 \begin{equation}
 \text{WADD}^{\bar{G}}(\bar{\tau}_c^*)=\sup_{G \in \mathcal{G}} \text{WADD}^{{G}}(\bar{\tau}_c^*).
	\end{equation}
Now, if $\tau$ is any stopping rule satisfying the mean time to false alarm constraint of $1/\alpha$, then since $\bar{\tau}_c^*$ is the optimal test for the LFL $\bar{G}$, we have 
	\begin{equation}
		\begin{split}
			\sup_{G \in \mathcal{G}} \text{WADD}^{{G}}({\tau})
   \geq \text{WADD}^{\bar{G}}({\tau})
   &\geq \text{WADD}^{\bar{G}}(\bar{\tau}_c^*) (1+\rc{o^*(1)})\\
			&= \sup_{G \in \mathcal{G}} \text{WADD}^{{G}}(\bar{\tau}_c^*)(1+\rc{o^*(1)}).
		\end{split}
	\end{equation}
 \rc{Here, as in the proof of Theorem~\ref{thm:LFLRobust}, 
         the term $o^*(1)$ is defined as }
        \begin{equation*}
            o^*(1) = 
            \begin{cases}
                0, &\quad \text{if $\bar{\tau}_c^*$ is exactly optimal,}\\
                o(1), &\quad \text{if $\bar{\tau}_c^*$ is asymptotically optimal}.\\
            \end{cases}
        \end{equation*}
       \rcn{Thus, if the stopping rule $\bar{\tau}_c^*$ in \eqref{eq:LFLCUSUM} is exactly optimal for \eqref{eq:QCDproblem2} when the post-change law is LFL, namely $\{\bar{G}\}$, then the $o^*(1)$ terms will be identically zero, and we will get exact robust optimality of $\bar{\tau}_c^*$ for the problem in \eqref{eq:robustProb}. If the rule $\bar{\tau}_c^*$ is asymptotically optimal for $\bar{G}$, then the $o^*(1)$ terms is the same as $o(1)$ which goes to zero in the limit $\alpha \to 0$. }
 \end{proof}


\section{Exact Robust Optimality in Non-Stationary Processes}
\label{sec:exactrobustopt}
In this section, we discuss two special cases where the robust optimality stated in the previous section is exact. 

When the LFL $\{\bar{g}_{n}\}$ satisfies the special condition that 
\begin{equation}
    \label{eq:robustexactiid}
    \bar{g}_{n} = \bar{g}, \quad \forall n \geq \nu,
\end{equation}
then the robust optimal algorithms reduce to the classical Shiryaev and CUSUM algorithm for the i.i.d. setting and hence are exactly optimal for their corresponding formulations. We note that this result is different from that studied in \cite{unni-etal-ieeeit-2011} because in \cite{unni-etal-ieeeit-2011} the post-change model is also unknown but assumed to be i.i.d., but here the post-change process is allowed to be any non-stationary process. We state this as a corollary.

\begin{corollary}
In Theorem~\ref{thm:LFLRobust} and Theorem~\ref{thm:LFLRobust_minimax}, if the LFL $\{\bar{g}_{n}\}$ satisfies the special condition that 
\begin{equation}
    \label{eq:robustexactiid_1}
    \bar{g}_{n} = \bar{g}, \quad \forall n \geq \nu,
\end{equation}
then the generalized Shiryaev algorithm designed using the LFL reduces to the classical Shiryaev algorithm with post-change density $\bar{g}$ with recursively implementable statistic
\begin{equation}
    \label{eq:ShirRecur_1}
    \rc{\bar{R}_n} = \frac{\rc{\bar{R}_{n-1}} + \rho}{1-\rho} \; \frac{\bar{g}(X_n)}{f(X_n)}, \quad \rc{\bar{R}_0}=0,
\end{equation}
and is exactly optimal for the robust problem stated in \eqref{eq:robustProb}. For the same reasons, the generalized CUSUM algorithm designed using the LFL reduces to the classical CUSUM algorithm with recursively implementable statistic 
\begin{equation}
\label{eq:CUSUMrobust}
    \rc{\bar{W}_n} = \left(\rc{\bar{W}_{n-1}} + \log \frac{\bar{g}(X_n)}{f(X_n)}\right)^+, \quad \quad \rc{\bar{W}_0}=0,
\end{equation}
and is exactly robust optimal for the problem stated in \eqref{eq:robustProbmini}.  

\end{corollary}

Another case where the generalized Shiryaev test is exactly optimal is when the LFL is an independent and periodically identically distributed (i.p.i.d.) process, i.e., 
$$
\bar{g}_{n+T} = \bar{g}_n, \quad \text{for some positive integer }T.
$$
For this case, it is shown in \cite{bane-tit-2021} that the generalized Shiryaev algorithm is exactly optimal for a periodic sequence of thresholds: \rc{$\bar A_{n, \alpha} = \bar A_{n+T, \alpha}$}, for every $n$. We state this result also as a corollary.
\begin{corollary}
    In Theorem~\ref{thm:LFLRobust}, if the LFL $\{\bar{g}_{n}\}$ is an i.p.i.d. process, then the generalized Shiryaev algorithm is exactly robust optimal. In this case, the statistic again has a recursive update 
    \begin{equation}
    \label{eq:ShirRecur_2}
    \rc{\bar{R}_n} = \frac{\rc{\bar{R}_{n-1}} + \rho}{1-\rho} \; \frac{\bar{g}_n(X_n)}{f(X_n)}, \quad \rc{\bar{R}_0}=0.
\end{equation}
\end{corollary}
Equivalent result for i.p.i.d. processes and the minimax setting is not yet available.

\section{General Conditions for Robust Asymptotic Optimality}
\label{sec:AsymptoticOpt}
When the LFL is not a fixed density, it is not known if, in general, the generalized Shiryaev or CUSUM algorithms are exactly optimal, \rc{i.e., when the constraint on the probability of a false alarm or the rate of false alarm is fixed to a finite value (However, see \cite{bane-tit-2021}, where the generalized Shiryev algorithm is shown to be exactly optimal in a non-stationary setting).} However, their asymptotic optimality (with the false alarm rate going to zero) has been extensively studied. One set of sufficient conditions on the asymptotic optimality of the generalized CUSUM algorithm essentially follows from the work of Lai (\cite{lai-ieeetit-1998}). \rc{Another set of (more general) sufficient conditions are obtained in \cite{liang2022quickest}. We discuss the conditions for the generalized CUSUM algorithm below for completeness.} For the general minimax and Bayesian theory (including theory for dependent data), we refer to \cite{lai-ieeetit-1998, tart-veer-siamtpa-2005, tart-niki-bass-2014, tart-book-2019, liang2022quickest}. 

\rc{It must be emphasized at this point that if the LFL is not stationary, then the statistic for the generalized Shiryaev and CUSUM algorithm 
cannot be calculated recursively (although window-limited versions have also been found to be optimal). }

 Define
$$
Z_{n} = \log \frac{\bar{g}_{n}(X_n)}{f(X_n)}
$$
to be the log-likelihood ratio at time $n$ when the change occurs at $\nu$ for the LFL. 
\begin{theorem}[\cite{brucks2023modeling},\cite{lai-ieeetit-1998}]
	\label{thm:modifiedconds}
	\begin{enumerate}
		\item Let there exist a positive number $I$ such that the log likelihood ratios $\{Z_{n}\}$ satisfy the following  condition: 
		\begin{equation}
			\label{eq:Znnu_LB}
			\begin{split}
				\lim_{n \to \infty} \; \sup_{\nu \geq 1} \; \esssup \mathsf{P}_\nu^{\bar{G}} &\left(\max_{t \leq n} \sum_{i = \nu }^{\nu + t} Z_{i} \geq I(1+\delta)n \; \bigg| \; X_1, \dots, X_{\nu-1}\right) = 0.
			\end{split}
		\end{equation}
		Then, we have the universal lower bound as $\gamma \to \infty$, 
		\begin{equation}
			\begin{split}
				\min_{\tau} \; \sup_{\nu \geq 1} \; &\esssup \; \mathsf{E}_\nu^{\bar{G}}[(\tau - \nu + 1)^+| X_1, \dots, X_{\nu-1}] \quad \geq \quad \frac{\log \gamma}{I} (1+o(1)). 
			\end{split}
		\end{equation}
		Here the minimum over $\tau$ is over those stopping times satisfying $\Expect_\infty[\tau] \geq \gamma$. 
		\item The generalized CUSUM algorithm
		$$
		\tau_{c}^* = \min\left\{n \geq 1: \max_{1 \leq k \leq n} \sum_{i=k}^n Z_{i} \geq \log(\gamma)\right\},
		$$
		satisfies
		$$
		\mathsf{E}_\infty[\tau_{c}^*] \geq \gamma.
		$$
		\item Furthermore, if the log-likelihood ratios $\{Z_{n}\}$ also satisfy 
		\begin{equation}
			\label{eq:Znnu_UB}
			\begin{split}
				\lim_{n \to \infty} \; \sup_{k \geq \nu \geq 1} \; \esssup \mathsf{P}_\nu^{\bar{G}} & \left(\frac{1}{n}\sum_{i = k }^{k +n} Z_{i} \leq I - \delta \; \bigg| \; X_1, \dots, X_{k-1}\right) = 0.
			\end{split}
		\end{equation}
		Then as $\gamma \to \infty$, $\tau_{c}^*$ achieves the lower bound:
		\begin{equation}
			\begin{split}
				\sup_{\nu \geq 1} \; \esssup \; \mathsf{E}_\nu^{\bar{G}} &[(\tau_{c}^* - \nu + 1)^+| X_1, \dots, X_{\nu-1}]  \leq  \frac{\log \gamma}{I}(1+o(1)), \quad \gamma \to \infty. 
			\end{split}
		\end{equation}
	\end{enumerate}
\end{theorem}

This theorem establishes the fact that if the conditions in the above theorem are satisfied, then the generalized CUSUM algorithm is minimax asymptotically robust and optimal for the problem in \eqref{eq:robustProbmini}. A similar result can also be stated for the generalized Shiryaev algorithm (\cite{tart-veer-siamtpa-2005, tart-niki-bass-2014, tart-book-2019}). 
When we have
\begin{equation}
    \label{eq:robustexactiid_2}
    \bar{g}_{n} = \bar{g}, \quad \forall n \geq \nu,
\end{equation}
the above theorem provides a performance analysis of the robust CUSUM algorithm. In this case, it can be shown that $I = D(\bar{g} \; \| \; f)$ and choosing threshold $\rc{\bar A} = \log \gamma$ ensures
\begin{equation}
    \label{eq:CUSUMMFA}
    \Expect_\infty[\tau_c^*] \geq \gamma.
\end{equation}
In addition, as $\rc{\bar A} = \log \gamma \to \infty$, we have
\begin{equation}
    \label{eq:CUSUMEDD}
    \Expect_1[\tau_c^*]  = \frac{\log(\gamma)}{D(\bar{g} \; \| \; f)}(1+o(1)).
\end{equation}

\rc{Recently, sufficient conditions more general than those discussed in Theorem~\ref{thm:modifiedconds} above for asymptotic optimality of the generalized CUSUM algorithm have been obtained in \cite{liang2022quickest}. Although the conditions given below are more general than those in Theorem~\ref{thm:modifiedconds}, we produce both of them here because the conditions in Theorem~\ref{thm:modifiedconds} are classical and have also been generalized to non-i.i.d. data. 
To discuss the sufficient conditions from \cite{liang2022quickest}, we first define the following growth function:
\begin{align}
\label{eqn:growth function}
    g_{\nu}(n) = \sum_{i = \nu}^{\nu + n - 1}\Expect_{\nu}\left[Z_{i}\right], \quad \forall n \geq 1.
\end{align}
Assume that for each $x > 0$, 
\begin{align}
    g^{-1}(x) \coloneqq \sup_{\nu \geq 1}g_{\nu}^{-1}(x)
\end{align}
exists. Note that both $g_{\nu}^{-1}$ and $g^{-1}$ are increasing and continuous.
The key assumption on $g^{-1}(x)$ guaranteeing the asymptotic optimality is
\begin{align}
    \log g^{-1}(x) = o(x), \quad \text{as }x \to \infty.
\end{align}
In this general theory, the condition \eqref{eq:Znnu_LB} is generalized to (for independent data)
\begin{align}
\label{eqn:12-Liang}
    \lim_{n \to \infty} \; \sup_{\nu \geq 1} \;  \mathsf{P}_\nu &\left(\max_{t \leq n} \sum_{i = \nu }^{\nu + t-1} Z_{i} \geq (1+\delta) g_{\nu}(n) \right) = 0, \quad \forall \delta > 0,
\end{align}
and the condition \eqref{eq:Znnu_UB} is generalized to (for independent data)
\begin{align}
\label{eqn:13-Liang}
   \lim_{n \to \infty} \max_{1\leq \nu \leq t} \Prob_{\nu}\left(\sum_{i = t}^{t + n -1} Z_{i} \leq (1 - \delta)g_{\nu}(n)\right) = 0, \quad \forall \delta \in  (0, 1).
\end{align}
Under these new conditions, it is shown that the optimal performance is given by
\begin{equation}
			\begin{split}
				\sup_{\nu \geq 1} \; \esssup \; \mathsf{E}_\nu^{\bar{G}} &[(\tau_{c}^* - \nu + 1)^+| X_1, \dots, X_{\nu-1}]  =  g^{-1}(\log \gamma)(1+o(1)), \quad \gamma \to \infty. 
			\end{split}
		\end{equation}
Thus, if the growth rate of KL-divergences $g_{\nu}(n)$ is exponential in $n$, then the delay will be proportional to $\log (\log (\gamma))$. }

We now give two examples of non-stationary processes for which the above conditions are satisfied.

\begin{example}[{I.P.I.D. Process}]
    If the LFL is an i.p.i.d. process, i.e., there exists an integer $T$ such that
    $$
    \bar g_{n+T} = \bar g_n, \quad \forall n, 
    $$
    then the conditions of Theorem~\ref{thm:modifiedconds} are satisfied with
    $$
    I = \frac{1}{T} \sum_{n=1}^T D(\bar g_n \; \| \; f). 
    $$
    See (\cite{bane-tit-2021, oleyaeimotlagh2023quickest, bane-elseviersp-2024}). 
\end{example}

\begin{example}[{MLR Order Processes}]
\label{exam:MLRorder}
     Let the LFL $\{\bar{g}_n\}$ be such that the densities $\{\bar{g}_n\}$ are increasing in MLR order: for all $n \geq 1$, 
    $$
    \frac{\bar{g}_{n+1}(x)}{\bar{g}_n(x)} \quad \uparrow \quad x.
    $$
    That is, the likelihood ratio between $\bar g_{n+1}$ and $\bar g_n$ is monotonically increasing in $x$, for every $n$, 
    then the conditions of Theorem~\ref{thm:modifiedconds} are satisfied with
    $$
    I = \lim_{N \to \infty} \frac{1}{N} \sum_{n=1}^N D(\bar{g}_n \; \| \; f),
    $$ 
    provided the above limit exists (\cite{brucks2023modeling}). 
\end{example}

\begin{example}[\rc{Processes with Uncontrolled Divergence Growth}]
     Consider the LFL in Example~\ref{exam:MLRorder} above. Let us assume now that $I = \lim_{N \to \infty} \frac{1}{N} \sum_{n=1}^N D(\bar{g}_n \; \| \; f) = \infty$. This means the rate of growth of KL divergences is superlinear. In this case, however, if the densities $\{\bar{g}_n\}$ satisfy \eqref{eqn:12-Liang} and \eqref{eqn:13-Liang}, then the generalized CUSUM with LFL is asymptotically optimal. We refer the reader to \cite{liang2022quickest} for details.


\end{example}




\section{Examples of Least Favorable Laws}
\label{sec:LFDexamples}
In this section, we provide examples of LFL from Gaussian and Poisson families. We use the following simple lemma in these examples. Its proof can be found, for example, in \cite{krishnamurthy2016partially}. We provide the proof for completeness. 
\begin{lemma}
\label{lem:MLRorder}
    Let $f$ and $g$ be two probability density functions such that $f$ dominated $g$ in monotone likelihood ratio (MLR) order:
    $$
    \frac{f(x)}{g(x)} \; \; \uparrow \; \; x. 
    $$
    Then, $f$ also dominates $g$ in stochastic order, i.e.,
    $$
    \int_{x}^\infty f(y) dy \geq \int_x^\infty g(y) dy, \quad \forall x. 
    $$
\end{lemma}
\begin{proof}
    Define, $t = \sup \left\{x: \frac{f(x)}{g(x)} \leq 1\right\}$.
    Then, when $x \leq t$, 
    $$
    \int_{-\infty}^x f(y) dy \leq  \int_{-\infty}^x g(y) dy,
    $$
    implying $
    \int_{x}^\infty f(y) dy \geq \int_x^\infty g(y) dy.
    $
    On the other hand, if $x > t$, then 
    $$
    \int_{x}^\infty f(y) dy =\int_{x}^\infty \frac{f(y)}{g(y)}g(y) dy \geq \int_x^\infty g(y) dy.
    $$
    This proves the lemma. 
\end{proof}
Note that the statement is valid even if the densities are with respect to a measure more general than the Lebesgue measure, including the counting measure. The latter implies that the MLR order for mass functions implies their stochastic order (this can be proved by replacing integrals with summations in the above proof). 

We now give two examples of MLR order densities. By the above lemma, they are also ordered by stochastic order. 
First, consider two Gaussian densities $g=\mathcal{N}(\mu_1, \sigma^2)$ and $f=\mathcal{N}(\mu_2, \sigma^2)$, with $\mu_2 > \mu_1$.  Then since
$$
\frac{f(x)}{g(x)} = \exp \left(\frac{\mu_{2}-\mu_1}{\sigma^2}x + \frac{\mu_1^2- \mu_{2}^2}{2 \sigma^2}\right),
$$
the ratio increases with $x$ because $(\mu_{2}-\mu_1) \geq  0$. 
Next, consider two Poisson mass functions
	$g=\text{Pois}(\lambda_1)$ and $f=\text{Pois}(\lambda_2)$, with $\lambda_2 > \lambda_1$. 
	Then since
	$$
	\frac{f(k)}{g(k)} = \left(\frac{\lambda_{2}}{\lambda_1}\right)^k e^{-\lambda_{2} + \lambda_1},
	$$
	the ratio increases with $k$ because $\frac{\lambda_{2}}{\lambda_1}\geq  1$.

We now give two examples of non-stationary processes with specified uncertainty classes. We then identify their LFLs. 

\begin{example}[Gaussian LFL]
\label{exam:GaussLFD}
    Let the pre-change density be given by
    $
    f = \mathcal{N}(0,1)
    $
    and the post-change densities are given by
    $$
    g_{n} = \mathcal{N}(\mu_{n}, 1). 
    $$
    Now, let for each $\nu$, $\mu_{n}$ is increasing in $n$. The means that $\{\mu_{n}\}$ are not known but are believed to satisfy for some known $\{\bar{\mu}_{n}\}$,
    $$
    \mu_{n} \geq \bar{\mu}_{n} > 0, \quad \forall n \geq \nu,
    $$
    Then, 
    \begin{align*}
\log \frac{\bar{g}_{n}(X)}{f(X)} =  \left(\bar{\mu}_{n}\; X - \frac{\bar{\mu}_{n}^2}{2}\right),
    \end{align*}
     and 
     \begin{align*}
        X \sim  \mathcal{N}(\bar{\mu}_{n}, 1) &\implies \log \frac{\bar{g}_{n}(X)}{f(X)}  \sim \mathcal{N} \left( \frac{\bar{\mu}_{n}^2}{2}, \; \bar{\mu}_{n}^2\right), \\
         X \sim  \mathcal{N}({\mu}_{n}, 1) &\implies \log \frac{\bar{g}_{n}(X)}{f(X)}  \sim \mathcal{N} \left(\bar{\mu}_{n}\; {\mu}_{n} - \frac{\bar{\mu}_{n}^2}{2}, \; \bar{\mu}_{n}^2\right).
     \end{align*}
      Since $\mu_{n} \geq \bar{\mu}_{n}$, we have
      $$
      \bar{\mu}_{n}\; {\mu}_{n} - \frac{\bar{\mu}_{n}^2}{2} \; \; \geq \; \; \frac{\bar{\mu}_{n}^2}{2}.
      $$
      Thus, $\mathcal{N} \left(\bar{\mu}_{n}\; {\mu}_{n} - \frac{\bar{\mu}_{n}^2}{2}, \; \bar{\mu}_{n}^2\right)$ dominates $\mathcal{N} \left( \frac{\bar{\mu}_{n}^2}{2}, \; \bar{\mu}_{n}^2\right)$ in MLR order (based on the discussion after Lemma~\ref{lem:MLRorder}) and hence in stochastic order (by Lemma~\ref{lem:MLRorder}). This shows that the post-change law made with the sequence of densities
      $\{\mathcal{N}(\bar{\mu}_{n}, 1)\}$ is the LFL.

\end{example}

We now give an example of LFL from the Poisson family of distributions. 
\begin{example}[Poisson LFL]
\label{exam:PoissonLFD}
    Let the pre-change density be given by
    $
    f = \text{Pois}(\lambda_0)
    $
    and the post-change densities are given by
    $$
    g_{n} = \text{Pois}(\lambda_{n}). 
    $$
    Now, let $\lambda_{n}$ is increasing in $n$. The means $\{\lambda_{n}\}$ are not known but are believed to satisfy for some known $\{\bar{\lambda}_{n}\}$,
    $$
    \lambda_{n} \geq \bar{\lambda}_{n} > \lambda_0, \quad \forall n \geq \nu,
    $$
    Then, 
    \begin{align*}
\log \frac{\bar{g}_{n}(X)}{f(X)} =   \log \left[\left(\frac{\bar{\lambda}_{n}}{\lambda_0}\right)^X e^{-\bar{\lambda}_{n} + \lambda_0}\right] = X \log \left(\frac{\bar{\lambda}_{n}}{\lambda_0}\right) -\bar{\lambda}_{n} + \lambda_0,
    \end{align*}
     and 
     \begin{align*}
        X \sim  \text{Pois}(\bar{\lambda}_{n}) &\implies \log \frac{\bar{g}_{n}(X)}{f(X)}  \sim \text{Pois} \left(\bar{\lambda}_{n} \log \left(\frac{\bar{\lambda}_{n}}{\lambda_0}\right) -\bar{\lambda}_{n} + \lambda_0\right), \\
         X \sim  \text{Pois}({\lambda}_{n}) &\implies \log \frac{\bar{g}_{n}(X)}{f(X)}  \sim \text{Pois} \left({\lambda}_{n} \log \left(\frac{\bar{\lambda}_{n}}{\lambda_0}\right) -\bar{\lambda}_{n} + \lambda_0\right) .
     \end{align*}
     We note that while the support of the random variable $\log \frac{\bar{g}_{n}(X)}{f(X)} $ is 
     $$
     \left\{k \log \left(\frac{\bar{\lambda}_{n}}{\lambda_0}\right) -\bar{\lambda}_{n} + \lambda_0\right\}, \quad k=0,1,2,3, \dots,
     $$
     its law is completely characterized by a Poisson distribution. This is what we represent using the $\sim$ sign above.  
      Since $\lambda_{n} \geq \bar{\lambda}_{n}$, we have
$$
\left({\lambda}_{n} \log \left(\frac{\bar{\lambda}_{n}}{\lambda_0}\right) -\bar{\lambda}_{n} + \lambda_0\right) \geq \left(\bar{\lambda}_{n} \log \left(\frac{\bar{\lambda}_{n}}{\lambda_0}\right) -\bar{\lambda}_{n} + \lambda_0\right).
$$
      Thus, $\text{Pois} \left({\lambda}_{n} \log \left(\frac{\bar{\lambda}_{n}}{\lambda_0}\right) -\bar{\lambda}_{n} + \lambda_0\right)$ dominates $\text{Pois} \left(\bar{\lambda}_{n} \log \left(\frac{\bar{\lambda}_{n}}{\lambda_0}\right) -\bar{\lambda}_{n} + \lambda_0\right)$ in MLR order (based on the discussion after Lemma~\ref{lem:MLRorder}) and hence in stochastic order (by Lemma~\ref{lem:MLRorder}). This shows that the post-change law made with the sequence of densities
      $\{\text{Pois}(\bar{\lambda}_{n})\}$ is the LFL.

\end{example}


\section{Numerical Results}
\label{sec:NumericalResults}
We show the effectiveness of the robust tests on simulated and real data. 
\subsection{Application to Simulated Gaussian and Poisson Data}
For the Gaussian case, we choose the following model:
\begin{equation}
    \begin{split}
        f = \mathcal{N}(0,1), \quad \quad g_{n} = \mathcal{N}(\mu_{n}, 1), \quad \quad \mu_{n} \geq 0.5.
    \end{split}
\end{equation}
From Example~\ref{exam:GaussLFD}, it follows that the LFL is $\mathcal{N}(0.5, 1)$. Thus, the CUSUM and Shiryaev test designed with pre-change density $f=\mathcal{N}(0,1)$ and post-change density $\mathcal{N}(0.5, 1)$ are robust optimal. In Fig.~\ref{fig:CUSUM} (Left), we have compared the robust CUSUM test with a randomly picked non-robust test designed using pre-change density $f=\mathcal{N}(0,1)$ and post-change density $\mathcal{N}(1.5, 1)$. The data samples after change for the simulations are generated using 
$\mathcal{N}(0.5, 1)$. This data generation corresponds to the worst-case scenario, and as expected, the worst-case performance of 
the robust test is better than that of a non-robust test.

For the Poisson case, we choose the following model:
\begin{equation}
    \begin{split}
        f = \text{Pois}(0.5), \quad \quad g_{n} = \text{Pois}(\lambda_{n}), \quad \quad \lambda_{n} \geq 0.8.
    \end{split}
\end{equation}
From Example~\ref{exam:PoissonLFD}, it follows that the LFL is $\text{Pois}(0.8)$. We compare it with a randomly chosen non-robust test with post-change law $\text{Pois}(1.5)$. The data for the Poisson case was generated using the worst-case scenario $\text{Pois}(0.8)$; See Fig.~\ref{fig:CUSUM} (Right). 

\begin{figure}[h]
    \centering
    \includegraphics[scale=0.45]{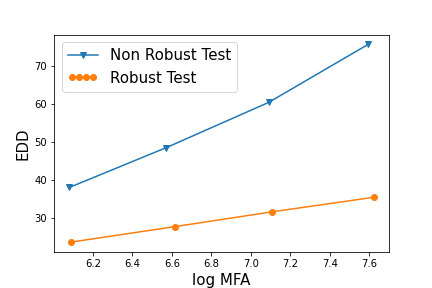}
    \includegraphics[scale=0.45]{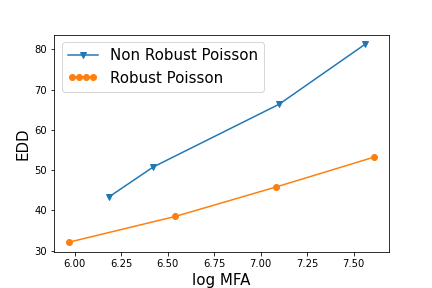}
    \caption{[Left] Comparison of robust CUSUM algorithm designed using pre-change density $f=\mathcal{N}(0,1)$ and post-change density $\mathcal{N}(0.5, 1)$ with a non-robust test designed using pre-change density $f=\mathcal{N}(0,1)$ and post-change density $\mathcal{N}(1.5, 1)$. The data samples are generated after change using $\mathcal{N}(0.5, 1)$. [Right] A similar plot for Poisson laws, with $f=\text{Pois}(0.5)$, LFL as $\text{Pois}(0.8)$, non-robust post-change as $\text{Pois}(1.5)$, and data samples are generated from $\text{Pois}(0.8)$. Here $\text{EDD}$ is used to denote $\Expect_1^{\bar{G}}[\tau_c^*-1]$ and 
   MFA $ = \Expect_{\infty}[\tau]$ is the mean time to a false alarm, where change time $\nu = \infty$ (no change).}
    \label{fig:CUSUM}
\end{figure}

\begin{figure}[h]
    \centering
    \includegraphics[scale=0.45]{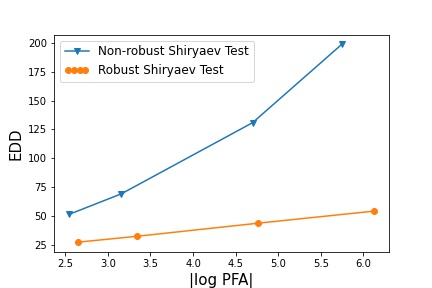}
    \includegraphics[scale=0.45]{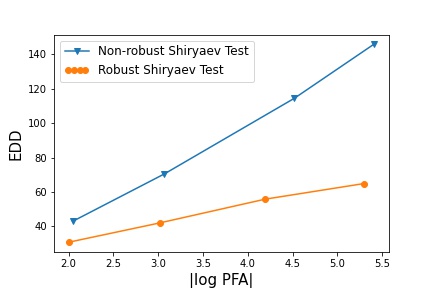}
    \caption{Comparisons for the robust Shiryaev test. These two figures were generated using the same setup used for generating Fig.~\ref{fig:CUSUM}. In addition, the change point is assumed to be a geometric random variable with parameter $0.01$. Here EDD represents $\Expect[(\tau^*-\nu)^+]$ and PFA $= \mathsf{P}^{\pi}(\tau < \nu)$ is the probability of false alarm.}
    \label{fig:Shiryaev}
\end{figure}
In Fig.~\ref{fig:Shiryaev}, we have shown the corresponding comparisons for the robust Shiryaev test using the same Gaussian and Poisson models as above. The change point is assumed to be a geometric random variable with parameter $0.01$.

\subsection{Applying Robust CUSUM Test to Pittsburgh Flight Data}
In this section, we apply the robust CUSUM test to detect the arrival of aircraft based on data collected around Pittsburgh-Butler Regional Airport (\cite{Patrikar2021}). The distance measurements for the last $100$ seconds for randomly selected $35$ flights are shown in Fig.~\ref{fig:exampleFlights}. We converted these distance measurements to signals by using the transformation $\frac{10}{\text{Distance}}$ (again see Fig.~\ref{fig:exampleFlights}). We then added $\mathcal{N}(0,1)$ noise after padding $100$ zeros 
at the beginning of the signals (see Fig.~\ref{fig:FlightSigNoise}). The noise is also added to simulate a scenario where an approaching enemy is being detected in a noisy environment. 


\begin{figure}[h]
    \centering
    \includegraphics[scale=0.45]{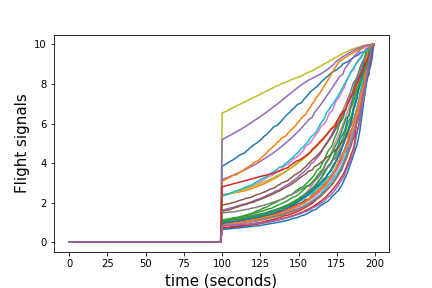}
    \includegraphics[scale=0.45]{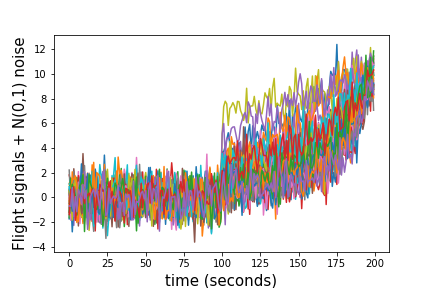}
    \caption{[Left] Flight signals for the last $100$ seconds of $35$ randomly chosen aircraft arriving at the Pittsburgh-Butler Regional Airport (\cite{Patrikar2021}). Signals are padded with zeros. The first $100$ seconds can be thought of as time before the aircraft appears in the sensor system. [Right] Flight signal corrupted by Gaussian noise, $\mathcal{N}(0,1)$. }
    \label{fig:FlightSigNoise}
\end{figure}


\begin{figure}[h]
    \centering
    \includegraphics[scale=0.45]{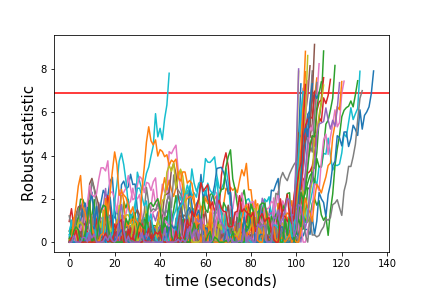}
    \includegraphics[scale=0.45]{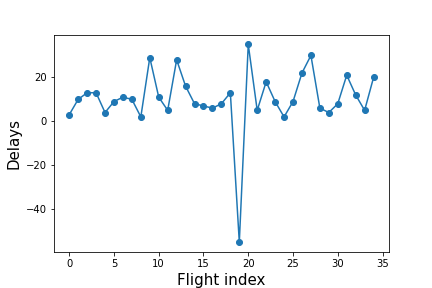}
    \caption{[Left] Robust CUSUM statistic obtained using $\mathcal{N}(0,1)$ as pre-change density and $\mathcal{N}(0.64,1)$ as post-change. The threshold $6.9 = \log(1000)$ is chosen to ensure the meantime to false alarm, MFA $> 1000$. [Right] The detection delay profile for the $35$ flights. We saw one false alarm (flight index $19$). The flight arrivals were detected with an average delay of 10.8 seconds over detection paths, i.e., after the false alarm path had been removed.}
    \label{fig:DelayPathProfile}
\end{figure}
Inspired by Example~\ref{exam:GaussLFD}, we design a robust CUSUM algorithm by selecting the pre-change density as $\mathcal{N}(0,1)$ and the post-change as $\mathcal{N}(0.64,1)$. Here $0.64$ is the minimum value of the flight signal across all $35$ flights. We choose the threshold $6.9 = \log(1000)$ to ensure that the mean time to a false alarm (MFA) is greater than $1000$; see \eqref{eq:CUSUMMFA}. The robust CUSUM statistics for all $35$ aircraft are plotted in Fig.~\ref{fig:DelayPathProfile} (Left) with the delay profile plotted on the right. The average detection delay over the delay paths was found to be $10.8$ seconds. 

\subsection{Applying Robust CUSUM Test to COVID Infection Data}
We now apply the robust CUSUM test to detect the onset of a pandemic. We downloaded the publicly available COVID infection dataset. We selected two U.S. counties, Allegheny and St Louis. The daily number of cases for the first $100$ days are plotted in Fig.~\ref{fig:exampleCOVID}. The infection rates between the two cities are of similar order since the two cities have similar population sizes. Similar to the Pittsburgh flight data in the previous section, we append the COVID data with zeros and add noise. Since the values are integers, we added $\text{Pois}(1)$ noise; see Fig.~\ref{fig:data_stat_Allegheny} (Left) and  
Fig.~\ref{fig:data_stat_StLouis} (Left). This data generation process simulates the scenario where it is required to detect the onset of a pandemic in the backdrop of daily infections due to other viruses or to detect the arrival of a new variant. 
Thus, the data represents the detection of deviation from a baseline. 

\begin{figure}[h]
    \centering
    \includegraphics[scale=0.45]{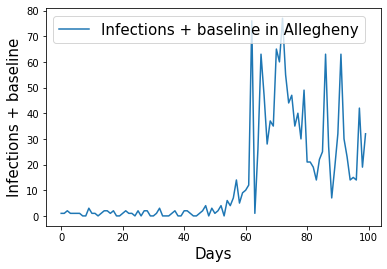}
    \includegraphics[scale=0.45]{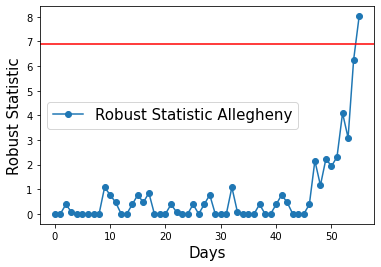}
    \caption{Noisy COVID infection data and robust statistic Allegheny county.}
    \label{fig:data_stat_Allegheny}
\end{figure}


\begin{figure}[h]
    \centering
    \includegraphics[scale=0.45]{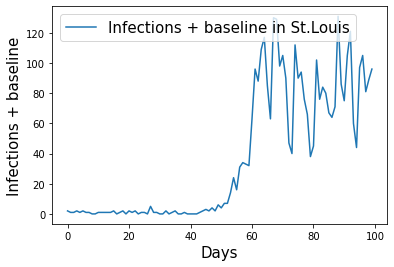}
    \includegraphics[scale=0.45]{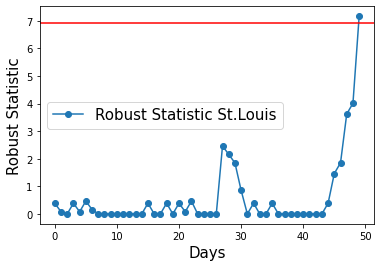}
    \caption{Noisy COVID infection data and robust statistic St. Louis county.}
    \label{fig:data_stat_StLouis}
\end{figure}


Since the actual number of infections is not known and can change over time, we take the robust approach and design the robust test with the LFL given by $\text{Pois}(2)$. The threshold is again chosen to be $6.9 = \log(1000)$ to guarantee a mean time to false alarm greater than $1000$ (see \eqref{eq:CUSUMMFA}). 
As seen in Fig.~\ref{fig:data_stat_Allegheny} (Right) and  
Fig.~\ref{fig:data_stat_StLouis} (Right), the 
robust CUSUM test detects the change quickly (within a week) since the infection rates started to become significant (in both counties) around the $50$th day.

\section{Conclusion}
We have developed optimal algorithms for the quickest detection of changes from an i.i.d. process to an independent non-stationary process. In most applications of quickest change detection, the post-change law is unknown, and learning it is much harder when the post-change process is also non-stationary. We have shown that if the post-change non-stationary family has a law that is least favorable, then the CUSUM or Shiryaev algorithms designed using the least favorable law (LFL) are robust and optimal. Note that while the LFL itself can be non-stationary, the most interesting case occurs when the LFL is stationary and also i.i.d. The robust optimal solution is then applied to COVID pandemic data and real flight data to detect anomalies. 




\bibliographystyle{tfcad}
\bibliography{TaposhQCD.bib}

\begin{thebibliography}{26}
\newcommand{\enquote}[1]{``#1''}
\providecommand{\natexlab}[1]{#1}
\providecommand{\url}[1]{\normalfont{#1}}
\providecommand{\urlprefix}{}

\bibitem[Banerjee, Gurram, and Whipps(2024)]{bane-elseviersp-2024}
Banerjee, Taposh, Prudhvi Gurram, and Gene Whipps. 2024. ``Minimax asymptotically optimal quickest change detection for statistically periodic data.'' \emph{Signal Processing} 215: 109290.

\bibitem[Banerjee, Gurram, and Whipps(2021)]{bane-tit-2021}
Banerjee, Taposh, Prudhvi Gurram, and Gene~T Whipps. 2021. ``A {Bayesian} theory of change detection in statistically periodic random processes.'' \emph{IEEE Transactions on Information Theory} 67 (4): 2562--2580.

\bibitem[Bansal and Papantoni-Kazakos(1986)]{bansal1986algorithm}
Bansal, R, and P~Papantoni-Kazakos. 1986. ``An algorithm for detecting a change in a stochastic process.'' \emph{IEEE Transactions on Information Theory} 32 (2): 227--235.

\bibitem[Brucks, Banerjee, and Mishra(2023)]{brucks2023modeling}
Brucks, Tim, Taposh Banerjee, and Rahul Mishra. 2023. ``Modeling and Quickest Detection of a Rapidly Approaching Object.'' \emph{Sequential Analysis} 42 (4): 387–403.

\bibitem[Krishnamurthy(2016)]{krishnamurthy2016partially}
Krishnamurthy, Vikram. 2016. \emph{Partially Observed Markov Decision Processes}. Cambridge University Press.

\bibitem[Lai(1998)]{lai-ieeetit-1998}
Lai, T.~L. 1998. ``Information Bounds and Quick Detection of Parameter Changes in Stochastic Systems.'' \emph{IEEE Transactions on Information Theory} 44 (7): 2917 --2929.

\bibitem[Liang, Tartakovsky, and Veeravalli(2022)]{liang2022quickest}
Liang, Yuchen, Alexander~G Tartakovsky, and Venugopal~V Veeravalli. 2022. ``Quickest change detection with non-stationary post-change observations.'' \emph{IEEE Transactions on Information Theory} 69 (5): 3400--3414.

\bibitem[Liang and Veeravalli(2022)]{liang2022non}
Liang, Yuchen, and Venugopal~V Veeravalli. 2022. ``Non-Parametric Quickest Mean-Change Detection.'' \emph{IEEE Transactions on Information Theory} 68 (12): 8040--8052.

\bibitem[Lorden(1971)]{lord-amstat-1971}
Lorden, G. 1971. ``Procedures for Reacting to a Change in Distribution.'' \emph{Annals of Mathematical Statistics} 42 (6): 1897--1908.

\bibitem[Molloy and Ford(2017)]{molloy2017misspecified}
Molloy, Timothy~L, and Jason~J Ford. 2017. ``Misspecified and asymptotically minimax robust quickest change detection.'' \emph{IEEE Transactions on Signal Processing} 65 (21): 5730--5742.

\bibitem[Molloy and Ford(2018)]{molloy2018minimax}
Molloy, Timothy~L, and Jason~J Ford. 2018. ``Minimax robust quickest change detection in systems and signals with unknown transients.'' \emph{IEEE Transactions on Automatic Control} 64 (7): 2976--2982.

\bibitem[Moustakides(1986)]{mous-astat-1986}
Moustakides, G.~V. 1986. ``Optimal Stopping Times for Detecting Changes in Distributions.'' \emph{Annals of Statistics} 14 (4): 1379--1387.

\bibitem[Oleyaeimotlagh et~al.(2023)]{oleyaeimotlagh2023quickest}
Oleyaeimotlagh, Yousef, Taposh Banerjee, Ahmad Taha, and Eugene John. 2023. ``Quickest Change Detection in Statistically Periodic Processes with Unknown Post-Change Distribution.'' \emph{Sequential Analysis} 42 (4): 404–437.

\bibitem[Page(1954)]{page-biometrica-1954}
Page, E.~S. 1954. ``Continuous Inspection Schemes.'' \emph{Biometrika} 41 (1/2): 100--115.

\bibitem[Patrikar et~al.(2021)]{Patrikar2021}
Patrikar, Jay, Brady Moon, Sourish Ghosh, Jean Oh, and Sebastian Scherer. 2021. ``{TrajAir: A General Aviation Trajectory Dataset}.''   \urlprefix\url{https://kilthub.cmu.edu/articles/dataset/TrajAir_A_General_Aviation_Trajectory_Dataset/14866251}.

\bibitem[Pollak(1985)]{poll-astat-1985}
Pollak, M. 1985. ``Optimal Detection of a Change in Distribution.'' \emph{Annals of Statistics} 13 (1): 206--227.

\bibitem[Pollak(1987)]{poll-astat-1987}
Pollak, M. 1987. ``Average Run Lengths of an Optimal Method of Detecting a Change in Distribution.'' \emph{Annals of Statistics} 15 (2): 749--779.

\bibitem[Poor and Hadjiliadis(2009)]{poor-hadj-qcd-book-2009}
Poor, H.~V., and O.~Hadjiliadis. 2009. \emph{Quickest detection}. Cambridge University Press.

\bibitem[Shiryaev(1963)]{shir-siamtpa-1963}
Shiryaev, A.~N. 1963. ``On Optimum Methods in Quickest Detection Problems.'' \emph{Theory of Probability and Its Applications} 8: 22--46.

\bibitem[Shumway and Stoffer(2010)]{shumway2010time}
Shumway, Robert~H, and David~S Stoffer. 2010. \emph{Time series analysis and its applications: with R examples}. Springer Science \& Business Media.

\bibitem[Tartakovsky, Nikiforov, and Basseville(2014)]{tart-niki-bass-2014}
Tartakovsky, A.~G., I.~V. Nikiforov, and M.~Basseville. 2014. \emph{Sequential Analysis: {Hypothesis} Testing and Change-Point Detection}. Statistics. CRC Press.

\bibitem[Tartakovsky and Veeravalli(2005)]{tart-veer-siamtpa-2005}
Tartakovsky, A.~G., and V.~V. Veeravalli. 2005. ``General Asymptotic {Bayesian} Theory of Quickest Change Detection.'' \emph{Theory of Probability and its Applications} 49 (3): 458--497.

\bibitem[Tartakovsky(2019)]{tart-book-2019}
Tartakovsky, Alexander. 2019. \emph{Sequential change detection and hypothesis testing: general non-iid stochastic models and asymptotically optimal rules}. CRC Press.

\bibitem[Unnikrishnan, Veeravalli, and Meyn(2011)]{unni-etal-ieeeit-2011}
Unnikrishnan, J., V.~V. Veeravalli, and S.~P. Meyn. 2011. ``Minimax Robust Quickest Change Detection.'' \emph{IEEE Transactions on Information Theory} 57 (3): 1604 --1614.

\bibitem[Veeravalli and Banerjee(2014)]{veer-bane-elsevierbook-2013}
Veeravalli, V.~V., and T.~Banerjee. 2014. \emph{Quickest Change Detection}. Academic Press Library in Signal Processing: Volume 3 -- Array and Statistical Signal Processing.

\bibitem[Xie, Liang, and Veeravalli(2023)]{xie2023distributionally}
Xie, Liyan, Yuchen Liang, and Venugopal~V Veeravalli. 2023. ``Distributionally Robust Quickest Change Detection using Wasserstein Uncertainty Sets.'' \emph{arXiv preprint arXiv:2309.16171} .

\end{thebibliography}

\end{document}